\documentclass{amsart}

\usepackage{xcolor}
\definecolor{linkblue}{HTML}{0000EE}

\usepackage
[
colorlinks=true,
linkcolor=linkblue,
anchorcolor=linkblue,
citecolor=linkblue,
urlcolor=linkblue,
plainpages=false,
pdfpagelabels
]{hyperref}

\usepackage{graphicx}
\usepackage{epstopdf}
\usepackage{amsmath}
\usepackage{amssymb}
\usepackage{xspace}
\usepackage{pb-diagram}
\usepackage{enumerate}
\usepackage{paralist}
\usepackage{thmtools}
\usepackage{nameref}
\usepackage[capitalize]{cleveref}
\usepackage{caption}

\usepackage[utf8]{inputenc}

\usepackage{lmodern}
\usepackage[T1]{fontenc}

\usepackage{tikz,tikz-cd}
\usetikzlibrary{matrix,arrows,positioning,calc,decorations.markings}

\usepackage[numbers,sort,compress]{natbib}
\usepackage[marginpar]{todo}

\usepackage{flushend}
\usepackage{mathtools}

\makeatletter
\renewcommand{\@tododisplay}[1]{%
\marginpar{#1}%
}
\renewcommand\@displaytodo[2][\todomark]{%
\@tododisplay{{\todoformat #1~(\ref{todolbl:\thetodo})}}%
\footnote[\thetodo]{\todoformat #1:~#2}%
\global\@todotoks\expandafter{\the\@todotoks\todoitem{#1}{#2}}%
\@todotrue%
}%
\renewcommand\todomark{todo}
\makeatother

\newcommand {\mm}[1] {\ensuremath{#1}}

\newcommand {\scalprod}[2] {{\langle #1 , #2 \rangle}}

\newcommand{\ignore}[1]{}

\makeatletter
\long\def\@makecaption#1#2{%
  \vskip\abovecaptionskip
  \sbox\@tempboxa{\small #1: #2}%
  \ifdim \wd\@tempboxa >\hsize
    \small #1: #2\par
  \else
    \global \@minipagefalse
    \hb@xt@\hsize{\hfil\box\@tempboxa\hfil}%
  \fi
  \vskip\belowcaptionskip}
\makeatother

\newcommand{\Rspace}        {\mm{{\mathbb R}}}

\newcommand{\sVdom}[3]      {\mm{{\rm Vor}_{#1}({#2},{#3})}}
\newcommand{\Ball}[2]       {\mm{B_{#1}{({#2})}}}

\newcommand{\sDel}[3]       {\mm{{\rm Del}_{#1}({#2},{#3})}}
\newcommand{\sDelTri}[2]    {\mm{{\rm Del}({#1},{#2})}}
\newcommand{\DelTri}[1]     {\mm{{\rm Del}({#1})}}
\newcommand{\Hasse}[1]      {\mm{{\mathcal H}{({#1})}}}
\newcommand{\Digraph}       {\mm{{\mathcal G}}}
\newcommand{\Singular}[2]   {\mm{{\rm Sing}_{#1}{({#2})}}}

\newcommand{\DSphere}[2]    {\mm{S{({#1},{#2})}}}
\newcommand{\Dsqradius}[2]  {\mm{s{({#1},{#2})}}}
\newcommand{\Incl}[1]       {\mm{{\rm Incl\,}{#1}}}
\newcommand{\Excl}[1]       {\mm{{\rm Excl\,}{#1}}}

\newcommand{\dsc}[1]        {\mm{{\downarrow}\,{#1}}}
\newcommand{\norm}[1]       {\mm{\|{#1}\|}}
\newcommand{\Edist}[2]      {\mm{d({#1},{#2})}}

\DeclareMathOperator{\DelOp}     {Del}
\DeclareMathOperator{\DelCechOp} {Del\v{C}ech}
\DeclareMathOperator{\CechOp}    {\v{C}ech}
\DeclareMathOperator{\WrapOp}    {Wrap}

\newcommand{\Cech}[2]       {\mm{\CechOp}_{#1}({#2})}
\newcommand{\DCech}[2]      {\mm{\DelCechOp}_{#1}({#2})}

\newcommand{\Del}[2]        {\mm{\DelOp}_{#1}({#2})}
\newcommand{\Wrap}[2]       {\mm{\WrapOp}_{#1}({#2})}

\newcommand{\radfun}[1]        {\mm{s_{#1}}}

\newcommand\restr[2]{{\left.\kern-\nulldelimiterspace#1\right|_{#2}}}
\newcommand\opdef[1]{\expandafter\DeclareMathOperator\csname #1\endcsname{#1}}

\DeclareMathOperator{\Front} {Front}
\DeclareMathOperator{\Back} {Back}
\DeclareMathOperator{\On} {On}

\DeclareMathOperator{\dime} {dim}

\newcommand{\vor}[1]        {\mathop \mathrm{Vor}{({#1})}}

\newcommand{\later}[1]     {{%
}}

\newtheorem{theorem}{Theorem}[section]
\newtheorem*{theorem*}{Theorem}

\newtheorem{lemma}[theorem]{Lemma} 
\newtheorem{corollary}[theorem]{Corollary}
 
\theoremstyle{definition}
\newtheorem{definition}[theorem]{Definition}

\newcommand*{\myref}[1]{%
  \nameref*{#1}%
  ~\hyperref[{#1}]{%
  \ref*{#1}%
  }}
\newcommand*{\theref}[1]{%
  the \myref{#1}%
  }

\renewcommand{\paragraph}[1]{\subsection{#1}}

\title{The Morse Theory of \goodbreak\v{C}ech and Delaunay Complexes
       }
       
\author{
Ulrich Bauer
\and
Herbert Edelsbrunner%
}


\begin{document}

\maketitle

\begin{abstract}
  Given a finite set of points in $\Rspace^n$ and a radius parameter,
  we study the \v{C}ech, Delaunay--\v{C}ech, Delaunay (or alpha),
  and Wrap complexes in the light of generalized discrete Morse theory.
  Establishing the \v{C}ech and Delaunay complexes as sublevel sets
  of generalized discrete Morse functions, we prove that the four complexes
  are simple-homotopy equivalent by a sequence of simplicial collapses,
  which are explicitly described by a single discrete gradient field.
\end{abstract}

\section{Introduction}
\label{sec1}

The burgeoning field of \emph{topological data analysis} was born from
the idea that results in algebraic topology can be fruitfully applied
to timely challenges in
data analysis \cite{Dey1999Computational,Carlsson2009Topology}.
It rests on the time-tested method of modeling by abstraction,
which in this setting means that we interpret the data as a finite sample
of a topological space.
Since we are given the data -- but not the space -- we construct a family
of hypothetical spaces and gain insights into the data from
general topological properties of these spaces and the relationships
between them.

\begin{figure}[t]
  \centering 
  \vspace{0.1in}
  \includegraphics[width=.40\textwidth]{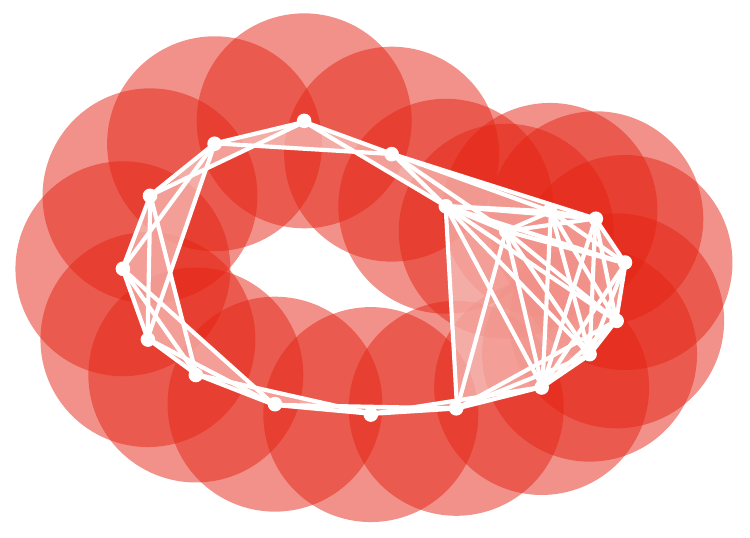} \quad%
  \includegraphics[width=.40\textwidth]{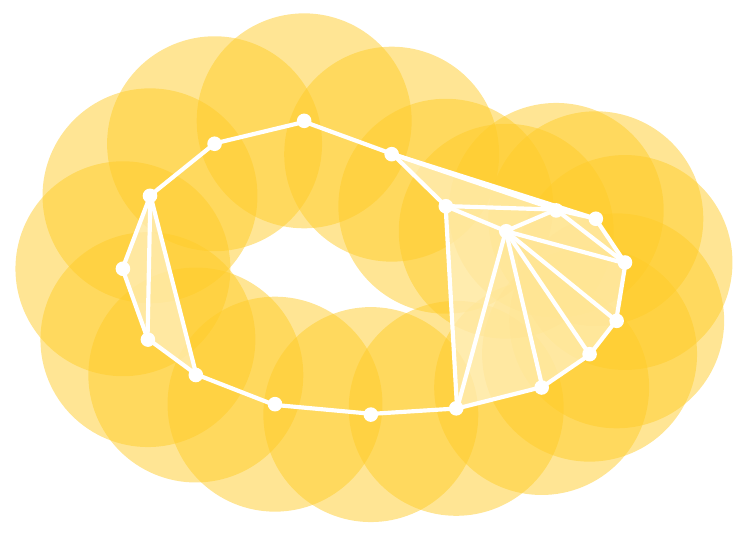}%
  \\
  \vspace{0.1in}
  \includegraphics[width=.40\textwidth]{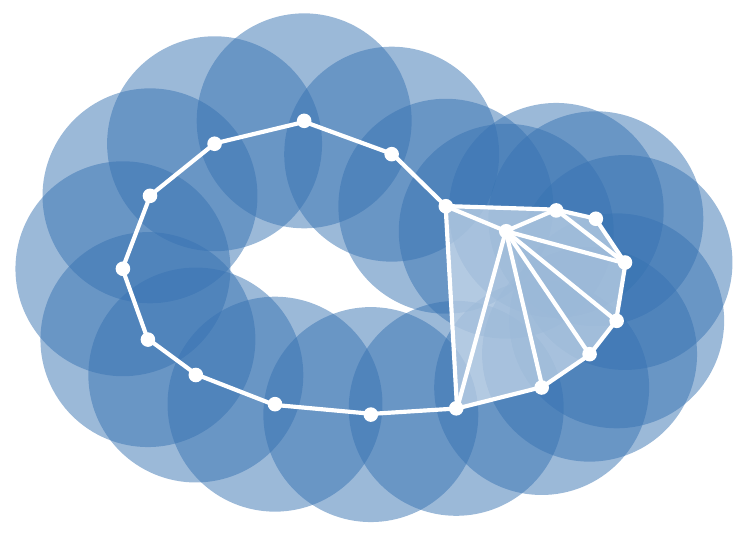} \quad%
  \includegraphics[width=.40\textwidth]{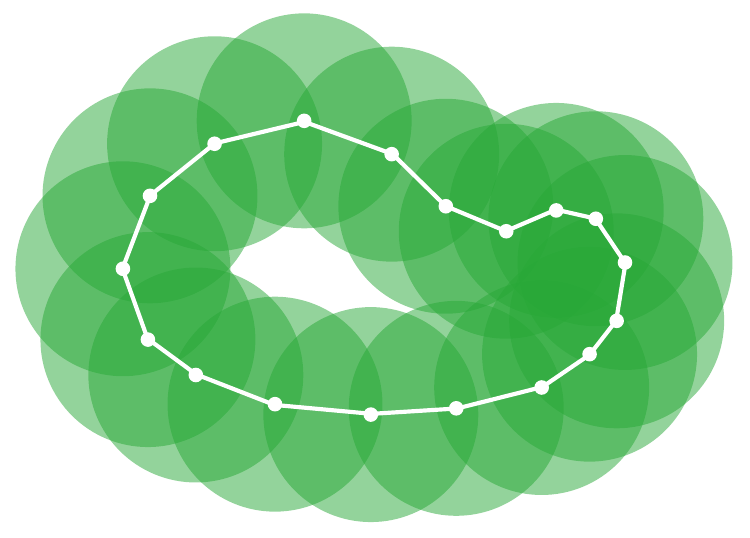}%
  \caption{The four different geometric complexes appearing in the collapsing
    sequence of the main theorem, drawn on top of the union of balls $B_r(X)$,
    to which they are all homotopy-equivalent.
    \emph{In sequence}: the high-dimensional \v{C}ech complex, $\Cech{r}{X}$,
    projected onto the plane,
    the Delaunay--\v{C}ech complex, $\DCech{r}{X}$,
    the Delaunay complex, $\Del{r}{X}$, and the Wrap complex, $\Wrap{r}{X}$.}
  \label{fig:complexes}
\end{figure}

\paragraph{Results}
Assuming the data consist of a set of points $X \subseteq \Rspace^n$,
we have several discrete geometric constructions at our disposal
that use the Euclidean distance and a scale parameter, $r$,
to convert the data into a filtration of spaces.
If the balls of radius~$r$ centered at $p+1$ data points have
a non-empty common intersections, then we may add the $p$-simplex
they span to the space representation.
This gives the \emph{\v{C}ech complex} for radius $r$,
denoted by $\Cech{r}{X}$, which is known to have the same homotopy type
as the union of balls, $\Ball{r}{X}$
\cite[Chapter III]{Edelsbrunner2010Computational}.
The construction of \v{C}ech complexes originates from the definitions
of cohomology theories for general topological spaces due to Alexandroff
and \v Cech \cite{Alexandroff1927Simpliziale,Cech1933Theorie,Cech1993General}.
Besides its foundational role in algebraic topology,
it was introduced by Carlsson and de Silva \cite{Silva2004Topological}
as a core construction in topological data analysis,
and has received significant interest in this area as well as
in stochastic geometry and topology \cite{Kahle2011Random}.

Alternatively, we may first intersect each ball with the Voronoi domain
of its center and then take the nerve.
This gives the \emph{alpha} or \emph{Delaunay complex} for radius $r$,
denoted by $\Del{r}{X}$, which embeds in $\Rspace^n$,
is a subcomplex of the \v{C}ech complex,
and has the same homotopy type \cite{Edelsbrunner1995Union}.
The Delaunay complex was first introduced under the name
\emph{$\alpha$-shape} as a construction for associating a geometric shape
to a finite set of points in the plane \cite{Edelsbrunner1983Shape}.
It has become a popular method both in computational geometry and topology;
see the survey \cite{Edelsbrunner2010Alpha}.
Not surprisingly, the \v{C}ech complex collapses to the Delaunay complex
for the same radius,
but this was an open question prior to this paper.

As a third option, we may collect all simplices in the \v{C}ech complex
that belong to the Delaunay triangulation.
This gives the \emph{Delaunay--\v{C}ech complex} for radius $r$,
denoted by $\DCech{r}{X}$.
This construction is a convenient alternative to Delaunay complexes
\cite{Chintakunta2013Topology,Pokorny2016Topological},
requiring only the computation of the Delaunay triangulation
and the smallest enclosing sphere of each Delaunay simplex.

As a fourth option, we may construct the \emph{Wrap complex},
denoted by $\Wrap{r}{X}$, which is a subcomplex of the
Delaunay triangulation \cite{Edelsbrunner2003Surface}.
Going beyond topological characterizations, it gives a geometric
description of the data and has been successfully employed within
commercial settings in software for surface reconstruction from point data,
serving as the foundation of the software package Geomagic Wrap®.
We extend the original $3$-dimensional notion to $\Rspace^n$
and introduce a dependence on a radius parameter, with the original 
definition corresponding to radius $\infty$.
Formulating the Wrap complex within Forman's discrete Morse theory
\cite{Forman1998Morse},
we answer an open question in \cite{Edelsbrunner2003Surface}.
Our main result is that the four complexes are related by simplicial
collapses, and that this property extend to natural generalizations
of the complexes to points with weights.
\begin{theorem*}[\v Cech--Delaunay Collapsing Theorem]
  Let $X$ be a finite set of possibly weighted points in general position
  in $\Rspace^n$.
  Then
  \begin{align*}
    \Cech{r}{X} \searrow \DCech{r}{X} \searrow \Del{r}{X} \searrow \Wrap{r}{X}
  \end{align*}
  for every $r \in \Rspace$.
\end{theorem*}
This establishes the \emph{simple-homotopy} equivalence of the four
different complexes, a particular type of homotopy equivalence that admits 
a purely combinatorial description in terms of elementary collapses
and expansions and is algebraically characterized by the vanishing
of Whitehead torsion \cite{Cohen1973Course}.

Our proofs are based on the insight that the Delaunay and \v{C}ech complexes arise
as sublevel sets of generalized discrete Morse functions with shared structural properties.
We refer to \cite{Forman1998Morse} for an introduction to discrete Morse theory,
and to \cite{Chari2000Discrete,Freij2009Equivariant} for the generalization of the discrete
gradient to allow for intervals larger than pairs.
Our constructions are elementary, using the radii of smallest enclosing
spheres and smallest empty circumspheres to define the
generalized Morse functions.
The geometric arguments are couched in the language of convex optimization,
in which there is little difference between ordinary points
and points with weights.
Indeed, all our results generalize to the weighted setting,
and thus relate to the theory of
power diagrams and regular triangulations \cite{Gelfand1994Discriminants}.

The common structure of Delaunay and \v{C}ech complexes leads naturally
to a generalization of the two constructions, which we call
the \emph{selective Delaunay complex}, $\sDel{r}{X}{E}$,
defined for a subset $E \subseteq X$ of \emph{excluded} points.
The construction of the selective Delaunay complex is based on smallest
enclosing spheres of subsets $Q \subseteq X$ whose interiors are empty
of the excluded points $E$.
The Delaunay and \v{C}ech complexes arise as special cases
with $E=X$ and $E=\emptyset$, respectively.
The main collapsing theorem is derived from the following statement,
which relates selective Delaunay complex with different excluded sets
through simplicial collapses:
\begin{theorem*}[Selective Delaunay Collapsing Theorem]
  Let $X$ be a finite set of possibly weighted points
  in general position in $\Rspace^n$, and let $E \subseteq F \subseteq X$.
  Then
  \begin{align*}
    \sDel{r}{X}{E} \searrow \sDel{r}{X}{E} \cap \sDelTri{X}{F}
                   \searrow \sDel{r}{X}{F}
  \end{align*}
  for every $r \in \Rspace$.
\end{theorem*}
It is worth mentioning that the concept of selective Delaunay complexes
enables the explicit construction of a sequence of maps in homology 
between Delaunay complexes $\Del{r}{X}$ and $\Del{r}{Y}$
that is equivalent to the maps in homology induced by the inclusions
$\Ball{r}{X}, \Ball{r}{Y} \hookrightarrow {\Ball{r}{X \cup Y}}$.

\paragraph{Related work}
A simple-homotopy version of the Nerve Theorem appears in \cite{Barmak2011Quillen}.
The author proves that a good cover of a simplicial complex
by subcomplexes has a nerve that is simple-homotopy equivalent to the complex.
Note however that this does not imply the stronger result that
the \v Cech complex collapses simplicially to the Delaunay complex.

The \emph{flow complex} introduced in \cite{Giesen2008Flow} is conceptually similar
to the Wrap complex.
Being based on the gradient flow of the distance function to a point set, this
construction is however less combinatorial.
In general, the flow complex is not a subcomplex of the Delaunay triangulation,
and its computation remains challenging~\cite{Cazals2008Robust}.
It has been shown to have the same homotopy type as the Delaunay complex
for the same radius \cite{Buchin2008Recursive}.
Our results imply that it has the same homotopy type as the Delaunay--\v{C}ech complex
and the Wrap complex, all for the same radius.

The structure of the generalized gradients has been described before
for the special case of \v{C}ech filtrations in \cite{Attali2013VietorisRips},
and for Delaunay filtrations in \cite{Edelsbrunner2003Surface,vanManen2014Power}.
Our Morse-theoretic treatment of selective Delaunay complexes systematically unifies and 
generalizes these results.
The continuous Morse theory for distance functions to finite point sets in Euclidean space 
has also been investigated in \cite{Siersma99Voronoi,Lieutier2004Any,Bobrowski2014Distance}. 
By the homotopy-equivalence of unions of balls, \v{C}ech complexes,
and Delaunay complexes, we obtain the same characterization of critical points
and the same statements about the change of homotopy type of sublevel sets
at critical values.
Our main interest is in the additional structure provided by the interval
partition of the discrete Morse function, and the explicit combinatorial
description of homotopy equivalences induced by a discrete gradient.

This paper extends the collapse of the Delaunay--\v{C}ech to
the Delaunay complex,
which is the main result in an early version of
the present paper \cite{Bauer2014Morse} and has been known
prior to that paper in $\Rspace^2$ only \cite{Chintakunta2013Topology}.
To further include the \v{C}ech complex in the collapsing sequence, we had
to substantially change the proof and
unify the treatment of \v{C}ech and Delaunay functions
in the framework of convex optimization.

The generalization of discrete Morse theory used in the present paper
was suggested in \cite{Freij2009Equivariant}.
The corresponding notion of collapses by intervals in the face poset
has been considered before in \cite{Wegner1975DCollapsing},
and collapses by even more general clusters of simplices have been
considered in \cite{Hersh2005Optimizing,Jonsson2008Simplicial}.
Another popular construction of a geometric complex is
the \emph{Vietoris--Rips complex} \cite{Vietoris1927Uber}.
The resulting filtrations however do not generally come from
a generalized Morse function.

\paragraph{Outline}
Section \ref{sec2} presents background material in combinatorial topology
and discrete Morse theory.
Section \ref{sec3} discusses the \v{C}ech and Delaunay complexes
in the context of a larger family of proximity complexes in Euclidean space.
Section \ref{sec4} generalizes from points to points with weights,
writing all conditions in the language of convex optimization.
Section \ref{sec5} proves the collapsing sequence.
Section \ref{sec6} describes consequences of the collapsing sequence.
Section \ref{sec7} concludes the paper.

\section{Background}
\label{sec2}

All complexes in this paper are simplicial, with vertices from a finite set in 
$\Rspace^n$.
Since we do not restrict the dimension of our simplices,
we will treat them as combinatorial rather than concrete geometric objects.
Assuming the reader is familiar with abstract simplicial complexes,
we give quick reviews of discrete Morse theory \cite{Forman1998Morse} and its generalization.

\paragraph{Discrete Morse theory.}
Given a finite set $X \subseteq \Rspace^n$,
we call a set $Q \subseteq X$ of $q+1$ points a \emph{$q$-simplex}.
Its \emph{dimension} is $q$,
its \emph{faces} are the subsets of $Q$,
and its \emph{facets} are the faces of dimension $q-1$.
A \emph{simplicial complex} is a collection of simplices, $K$,
that is closed under the face relation,
and its \emph{dimension} is the maximum dimension of any of its simplices.
The face relation defines a canonical partial order on $K$,
and the \emph{Hasse diagram}, denoted as $\Hasse{K}$,
is the transitive reduction of this order.
In other words, $\Hasse{K}$ is the directed acyclic graph whose nodes
are the simplices and whose arcs are the pairs $(P, Q)$ in which $P$
is a facet of~$Q$.
A \emph{discrete vector field} is a partition $V$ of $K$
into singleton sets $\{C\}$ and pairs $\{P, Q\}$ corresponding
to arcs $(P,Q)$ in the Hasse diagram.
Suppose now that there is a function $f \colon K \to \Rspace$
that satisfies $f(P) \leq f(Q)$ whenever $P$ is a face of $Q$,
with equality holding in this case iff $(P, Q)$ is a pair in $V$.
Then $f$ is called a \emph{discrete Morse function}
and $V$ is its \emph{discrete gradient}.
Indeed, the existence of $f$ attests for the acyclicity of the directed graph
obtained from the Hasse diagram by contracting the pairs in $V$.
A simplex that does not belong to any pair in $V$ is called a
\emph{critical simplex} and the corresponding value is a
\emph{critical value} of $f$.

To provide an intuition for the concept, we note that the pairs
in a discrete gradient correspond to elementary collapses
\cite[Chapter III]{Edelsbrunner2010Computational}, except that the lower-dimensional
simplex does not have to be free.
An elementary collapse can be realized continuously by
a deformation retraction.
This implies that if we can transform a simplicial complex, $K$,
to another, $K'$, using a sequence of such elementary collapses,
then the two complexes have the same homotopy type. In fact, the implied relation is slightly stronger, which is usually
expressed by saying that $K$ and $K'$ are
\emph{simple-homotopy equivalent} \cite{Cohen1973Course}.
We say that $K$ collapses onto $K'$ and write $K \searrow K'$.
A discrete gradient can encode a collapse \cite{Forman1998Morse}:
\begin{theorem}[Gradient Collapsing Theorem]
  \label{Gradient Collapsing Theorem}
  Let $K$ be a simplicial complex with a discrete gradient $V$,
  and let $K' \subseteq K$ be a subcomplex.
  If $K \setminus K'$ is a union of pairs in $V$, then $K \searrow K'$.
\end{theorem}
We say that the collapse from~$K$ to~$K'$ is \emph{induced by $V$.}

\paragraph{Generalized discrete Morse theory}
To generalize discrete Morse theory
we recall that an \emph{interval}
in the face relation of $K$ is a subset of the form
\begin{align}
  [P,R]  &=  \{ Q \mid P \subseteq Q \subseteq R \} .
\end{align}
The interval is non-empty iff $P$ is a face of $R$.
In this case, the interval contains both simplices
-- which may be the same --
and we refer to $P$ as its \emph{lower bound} and to $R$
as its \emph{upper bound}.
Borrowing from the nomenclature of \cite{Freij2009Equivariant}, we call a partition $W$
of $K$ into intervals a \emph{generalized discrete vector field}.
Indeed, a discrete vector field is the special case in which all intervals
are either singletons or pairs.
Suppose now that there is a function $f \colon K \to \Rspace$
that satisfies $f(P) \leq f(Q)$ whenever $P$ is a face of $Q$,
with equality holding in this case iff $P$ and $Q$ belong to a common interval in $W$.
Then $f$ is called a \emph{generalized discrete Morse function}
and $W$ is its \emph{generalized discrete gradient}.
If an interval contains only one simplex,
then we call the interval \emph{singular},
the simplex a \emph{critical simplex}, and the value of the simplex
a \emph{critical value} of $f$.

It is easy to see that for every generalized discrete gradient,
there is a discrete gradient that refines every non-singular interval
$[P, R]$ into pairs:
choose an arbitrary vertex $x \in R \setminus P$ and partition
$[P, R]$ into pairs $\{Q \setminus \{x\}, Q \cup \{x\}\}$
for all $Q \in [P, R]$.
We call this a \emph{vertex refinement}.
The generalized discrete gradient and its refinement
have the same critical simplices, implying that \theref{Gradient Collapsing Theorem}
also applies to generalized discrete gradients.
The refinement is in general not unique.
\begin{theorem}[Generalized Gradient Collapsing Theorem]
  \label{Generalized Gradient Collapsing Theorem}
  Let $K$ be a simplicial complex with a generalized discrete gradient $V$,
  and let $K' \subseteq K$ be a subcomplex.
  If $K \setminus K'$ is a union of non-singular intervals in $V$,
  then $K \searrow K'$.
\end{theorem}

\section{Proximity complexes}
\label{sec3}

We introduce \v{C}ech and Delaunay complexes as members of the larger
family of selective Delaunay complexes.
After writing these complexes as sublevel sets of real-valued functions,
we introduce the Delaunay--\v{C}ech complexes as subcomplexes
of the Delaunay triangulation.  

\paragraph{\v{C}ech complexes}
Write $\Edist{x}{y}$ for the Euclidean distance between $x, y \in \Rspace^n$.
For $r \geq 0$, let $B_r (x) = \{y \in \Rspace^n \mid \Edist{x}{y} \leq r\}$
be the closed ball of radius $r$ centered at $x \in X$.
The \emph{\v{C}ech complex} of a finite set $X \subseteq \Rspace^n$
for radius $r \geq 0$, 
\begin{align}
  \Cech{r}{X}  &=  \bigg\{ Q \subseteq X
                    \mid \bigcap_{x \in Q} B_r (x) \neq \emptyset \bigg\} ,
\end{align}
is isomorphic to the nerve of the collection of closed balls.
Therefore, by the Nerve Theorem \cite{Borsuk1948Imbedding},
it is homotopy equivalent to the union of the balls, $B_r (X) = \bigcup_{x \in X} B_r (x)$.
For sufficiently large radius, the \v{C}ech complex is the
full (abstract) simplex spanned by $X$,
which we denote as $\Delta (X) = 2^X \setminus \{ \emptyset \}$.
For $r \leq t$, we have $\Cech{r}{X} \subseteq \Cech{t}{X}$,
so the \v{C}ech complexes form a \emph{filtration} of $\Delta (X)$.

\paragraph{Delaunay complexes}
Let $X \subseteq \Rspace^n$ be again finite and $x \in X$.
The \emph{Voronoi domain} of $x$ with respect to $X$,
and the \emph{Voronoi ball} of $x$ with respect to $X$
for a radius $r \geq 0$ are
\begin{align}
  \sVdom{}{x}{X}  &=  \{ y \in \Rspace^n  \mid
    \Edist{y}{x} \leq \Edist{y}{p} \text{ for all } p \in X \},  \\
  \sVdom{r}{x}{X}  &=  B_r(x) \cap \vor{x,X} .
\end{align}
The \emph{Delaunay complex} of $X$ for radius $r \geq 0$,  
\begin{align}
  \Del{r}{X}    =  \bigg\{ Q \subseteq X
                    \mid \bigcap_{x \in Q} \sVdom{r}{x}{X}
                                              \neq \emptyset \bigg\} ,
\end{align}
often referred to as alpha complex,
is isomorphic to the nerve of the collection of Voronoi balls.
For sufficiently large $r$, this is the \emph{Delaunay triangulation}
of $X$, which we denote as $\DelTri{X} = \Del{\infty}{X}$.
For $r \leq t$, we have $\Del{r}{X} \subseteq \Del{t}{X}$,
so the Delaunay complexes form a filtration of the Delaunay triangulation.
If we assume that the points are in general position, then the
Delaunay triangulation is geometrically realized as a simplicial
complex in $\Rspace^n$.
While the assumption of general position is not required
in the abstract setting, we will need it in the construction
of generalized discrete Morse functions.

\paragraph{Delaunay--\v{C}ech complexes}
The \emph{Delaunay--\v{C}ech complex} for radius $r \geq 0$ is the restriction
of the \v{C}ech complex to the Delaunay triangulation.
It contains all simplices
in the Delaunay triangulation such that the balls of radius $r$
centered at the vertices have a non-empty common intersection:
\begin{align}
  \DCech{r}{X}  &=  \bigg\{ Q \in \DelTri{X} \mid
                      \bigcap_{x \in Q} B_r (x) \neq \emptyset \bigg\} .
\end{align}
Similar to the Delaunay complex,
we have $\DCech{r}{X} \subseteq \DCech{t}{X} \subseteq \DelTri{X}$
whenever $r \leq t$, so the Delaunay--\v{C}ech complexes also
form a filtration of the Delaunay triangulation.
While the Delaunay and the Delaunay--\v{C}ech complexes are similar,
they are not necessarily the same,
as illustrated in Figure \ref{fig:complexes}.
Instead of equality, we have
$\Del{r}{X} \subseteq \DCech{r}{X}$ for all $r$.
To see this, we just note that if the Voronoi balls have a non-empty common intersection,
then the balls have a non-empty common intersection
and so do the Voronoi domains.

\paragraph{Selective Delaunay complexes.}
The proof of the main result in this paper makes essential use
of a family of complexes that contain the Delaunay and the 
\v Cech complexes as extremal cases.
To introduce this family, let $X \subseteq \Rspace^n$ be a 
finite set, $E \subseteq X$ a subset,
and $r \geq 0$ a real number.
Note that for a point $x \in X$ we can also consider the \emph{Voronoi ball}
of $x$ with respect to a subset $E \subseteq X$ not necessarily
containing $x$.
Specifically, $\sVdom{r}{x}{E}$ is the set of points $y \in 
\Rspace^n$ whose distance to $x$ is bounded from above by 
$r$ and by the distances to the points in $E$.
The \emph{selective Delaunay complex} for $E \subseteq X$ and $r$
contains all simplices over $X$ whose vertices have Voronoi balls
for the subset $E$
with non-empty common intersection:
\begin{align}
  \sDel{r}{X}{E}  &=  \bigg\{ Q \subseteq X \mid
    \bigcap_{x \in Q} \sVdom{r}{x}{E} \neq \emptyset \bigg\} .
\end{align}
It is isomorphic to the nerve of the set of these Voronoi balls;
see Figure \ref{fig:selectiveDelaunay}.
\begin{figure}[htb]
  \centering 
  \vspace{0.2in}
  \includegraphics[height=.9in]{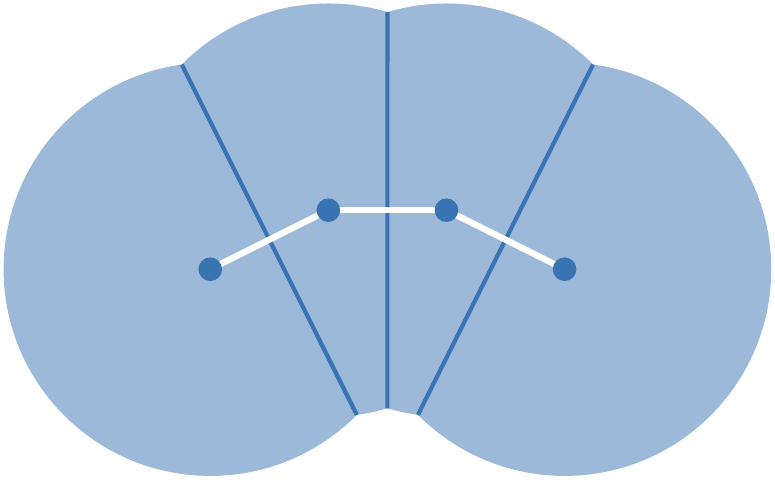}\quad
  \includegraphics[height=.9in]{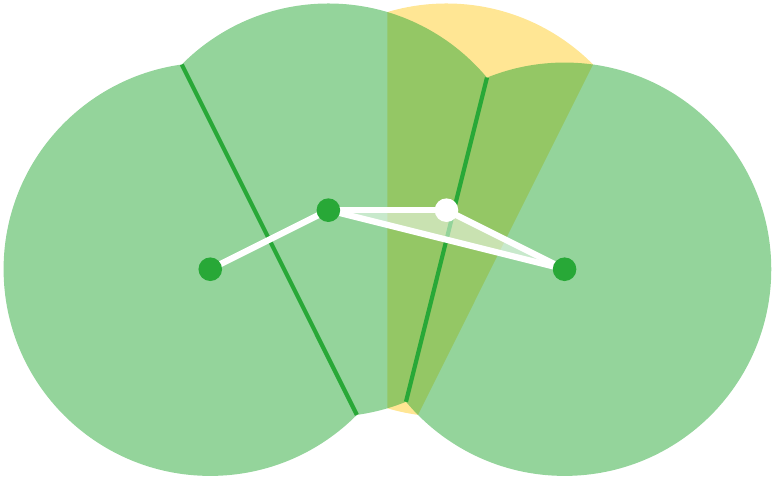}\quad
  \includegraphics[height=.9in]{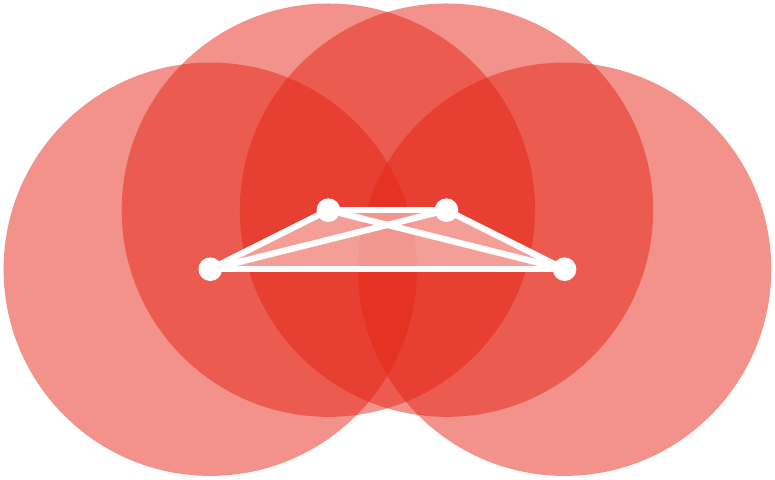}
  \caption{The Voronoi balls of a set $X$ of four points in $\Rspace^2$.
    In the \emph{middle}, three of the balls belong to the subset
    $E \subseteq X$ that impose constraints.
    The corresponding selective Delaunay complex, $\sDel{r}{X}{E}$,
    has four edges and one triangle.
    It properly contains the Delaunay complex, $\Del{r}{X}$,
    shown on the \emph{left}, and it is properly contained in the
    \v{C}ech complex, $\Cech{r}{X}$, shown on the \emph{right}.}
  \label{fig:selectiveDelaunay}
\end{figure}
While individual Voronoi balls depend on $E$, their union does not.
To see this, we note that $\sVdom{r}{x}{F} \subseteq \sVdom{r}{x}{E}$
whenever $E \subseteq F$ for all $x \in X$.
We get the largest Voronoi balls for $E = \emptyset$,
in which case each domain is a ball of radius $r$.
We get the smallest Voronoi balls for $E = X$,
in which case the Voronoi balls form a convex decomposition
of the union of balls.
Since the union does not depend on $E$, the Nerve Theorem \cite{Borsuk1948Imbedding} implies that 
for a given point set $X$ and radius $r$,
all selective Delaunay complexes have the same homotopy type.
Note also that $\sDel{r}{X}{F} \subseteq \sDel{t}{X}{E}$
whenever $r \leq t$ and $E \subseteq F$.
For $E = \emptyset$, the selective Delaunay complex is the \v{C}ech complex:
$\Cech{r}{X} = \sDel{r}{X}{\emptyset}$.
For $E = X$, the selective Delaunay complex is the Delaunay complex:
$\Del{r}{X} = \sDel{r}{X}{X}$.
In analogy to the Delaunay triangulation, we define 
$\sDel{}{X}{E}=\sDel{\infty}{X}{E}$.

\paragraph{Radius functions}
There is an equivalent, dual definition of selective Delaunay complexes,
which is natural from the point of view of discrete Morse theory
and will reveal important structural properties.
To state this definition, consider two point sets
$Q, E \subseteq X \subseteq \Rspace^n$.
We say an $(n-1)$-dimensional sphere $S$
in $\Rspace^n$ \emph{includes} $Q \subseteq X$
if all points of $Q$ lie on or inside $S$,
and it \emph{excludes} $E \subseteq X$ if all points of $E$ lie on or outside $S$.
If $Q$ and $E$ share points, then they necessarily lie on $S$.
The set of such spheres may be empty,
but if it is not, then we define $\DSphere{Q}{E}$ to be the smallest
such sphere, referring to it as the
\emph{Delaunay sphere} of $Q$ with respect to $E$,
and we write $\Dsqradius{Q}{E}$ for its squared radius.
The \emph{radius function} for $E$ maps each simplex
to the squared radius of the Delaunay sphere:
\begin{align}
  \radfun{E}  \colon \sDelTri{X}{E} \to \Rspace 
\end{align}
defined by $\radfun{E} (Q) = \Dsqradius{Q}{E}$.
Considering the special case $E = X$ of Delaunay complexes,
we call $\radfun{X}$ the \emph{Delaunay radius function} of $X$,
and $\DSphere{Q}{X}$ the \emph{Delaunay sphere} of $Q$,
which is commonly referred to as the \emph{smallest empty circumsphere} of $Q$.
Similarly, for the special case $E = \emptyset$ of \v Cech complexes,
we call $\radfun{\emptyset}$ the \emph{\v Cech radius function} of $X$,
and $\DSphere{Q}{\emptyset}$ the \emph{\v Cech sphere} of $Q$,
which is commonly referred to as the \emph{smallest enclosing sphere} of $Q$.

It is not difficult to see that $Q$ belongs to the selective Delaunay complex
for radius $r$ iff its value under the radius function exists
and does not exceed $r^2$.
For completeness, we present the formal claim with proof.
\begin{lemma}[Radius Function Lemma]
  \label{Radius Function Lemma}
  Let $X \subseteq \Rspace^n$ be finite, $E \subseteq X$, and $r \geq 0$.
  A simplex $Q \in \sDel{}{X}{E}$ belongs to $\sDel{r}{X}{E}$
  iff $\radfun{E} (Q) \leq r^2$.
\end{lemma}
\begin{proof}
  Suppose $Q \in \sDel{r}{X}{E}$,
  consider the Voronoi balls of its vertices with respect to $E$,
  and let $y$ be a common point of these balls.
  Let $\rho$ be the maximum distance between $y$ and any
  point in $Q \cap E$. By construction, the sphere with center~$y$ 
  and radius~$\rho$ includes $Q$ and excludes $E$. 
  We have $\radfun{E} (Q) \leq \rho^2 \leq r^2$
  since $\radfun{E} (Q)$ is the squared radius
  of the smallest such sphere.

  Conversely, if $\radfun{E} (Q) \leq r^2$, then the center of the smallest
  sphere that includes $Q$ and excludes $E$ belongs to the
  Voronoi ball of every point in $Q$ with respect to $E$,
  which implies $Q \in \sDel{r}{X}{E}$.
\end{proof}

We observed earlier that $\sDel{r}{X}{F} \subseteq \sDel{r}{X}{E}$
whenever $E \subseteq F$.
Correspondingly, we have $\radfun{E} (Q) \leq \radfun{F} (Q)$
whenever both are defined, which  is clear because $E$ generates fewer constraints than $F$
and therefore allows for a radius that is smaller than or equal
to the smallest radius we get for $F$.

\section{Convex Optimization}
\label{sec4}

Assuming the points are in general position, we prove that the
radius functions defined in Section \ref{sec3} are
generalized discrete Morse functions.
For this purpose, we translate the geometric constructions into
the language of convex optimization.
In this setting, there is little difference between
points and weighted points, so we generalize all results to weighted points.

\paragraph{Weighted points.}
Using \theref{Radius Function Lemma}, we can determine whether or not
a simplex $Q$ belongs to $\sDel{r}{X}{E}$ by solving a convex
optimization problem:
$\DSphere{Q}{E}$ is the sphere with center $z$ and radius $r \leq 0$ that
\begin{align}
  \underset{r,z}{\text{minimizes~~}} & r^2 \\
  \text{subject to~~} & \Edist{z}{q}^2 \leq r^2, ~~\forall q \in Q, \\
                    & \Edist{z}{e}^2 \geq r^2, ~~\forall e \in E.
\end{align}
We generalize the setting to allow for points $x \in \Rspace^n$
with weight $w_x \in \Rspace$, a concept well known
for Voronoi diagrams and Delaunay triangulations,
whose weighted versions are sometimes referred to as
\emph{power diagrams} \cite{Aurenhammer1987Power}
and \emph{regular triangulations} \cite{Gelfand1994Discriminants}.
To explain, we substitute $s = r^2$, allowing $s$ to be negative as well.
Appealing to geometric intuition, we speak of a sphere $S$ with
squared radius $s$ and center $z$ nonetheless.
We say $S$ \emph{includes} a point $x$ with weight $w_x$
if $\Edist{z}{x}^2 \leq s + w_x$,
and $S$~\emph{excludes}~$x$ if $\Edist{z}{x}^2 \geq s + w_x$.
Similarly, $x$ \emph{lies on} $S$ if it is simultaneously included and excluded.
With this extension of the relations,
we can read everything we said about spheres and points
as statements about spheres and weighted points.
To obtain an intuitive geometric interpretation of the weight,
we consider the sphere $S_x$ with center $x$ and positive
squared radius $w_x$.
The weighted point~$x$ then lies on~$S$ iff the two spheres~$S$ and $S_x$
intersect orthogonally.
Similarly, $S$ includes $x$ iff the two spheres are orthogonal
or closer to each other than orthogonal,
and $S$ excludes $x$ iff the two spheres are orthogonal
or further from each other than orthogonal.

With the new notation, we rewrite the convex optimization problem
so it applies to the more general, weighted setting:
\begin{align}
  \underset{s,z}{\text{minimize~~}} & s                        
  \label{eqn:G1} 
  \\
  \text{subject to~~}  & \Edist{z}{q}^2 \leq s + w_q, ~~\forall q \in Q,
  \label{eqn:G2} 
                                                               \\
                     & \Edist{z}{e}^2 \geq s + w_e, ~~\forall e \in E.
  \label{eqn:G3}
\end{align}
This effectively generalizes the notion of selective Delaunay complexes
from equal sized balls to sets in which balls can have different sizes.
Such more general data occurs in a number of applications,
including the modeling of biomolecules.

\paragraph{Karush--Kuhn--Tucker conditions.}
In the next step, we reformulate the optimization problem so that
the objective function is convex and the constraints are affine.
Before we get there, consider a general optimization problem:
\begin{align}
  \underset{y}{\text{minimize}} & ~~f(y) \\
  \text{subject to} & ~~g_j (y) \leq 0, ~~\forall j \in J,\\
  & ~~g_k (y) = 0, ~~\forall k \in K,\\
  & ~~g_l (y) \geq 0, ~~\forall l \in L,
\end{align}
in which $J, K, L$ are pairwise disjoint index sets.
Assuming $f$ is convex and the $g_i$ are affine,
the \emph{Karush--Kuhn--Tucker conditions} say that a feasible point $y$
is an optimal solution iff there exist coefficients $\lambda_i \in \Rspace$
for all $i \in I = J \cup K \cup L$ such that
\begin{align}
  \nabla f(y) + \sum_{i \in I} \lambda_i \nabla g_i (y)  &=  0 , \label{eqn:A1} \tag{stationarity}\\
  \lambda_i g_i (y)  &=  0, ~~\forall i \in I,              \label{eqn:A2} \tag{complementary slackness}\\
  \begin{split}
  \lambda_j  &\geq  0, ~~\forall j \in J,                    \\
  \lambda_l  &\leq  0, ~~\forall l \in L;                 
  \end{split} \label{eqn:A3}\tag{dual feasibility}
\end{align}
see \cite{Boyd2004Convex}.
In particular, necessity is provided by the linearity of the
constraints \cite[p.~226]{Boyd2004Convex},
while sufficiency is provided by the convexity of the objective function
and the inequality constraints \cite[p.~244]{Boyd2004Convex}.
To get our problem into this form, we introduce $a = \norm{z}^2 - s$,
write $y = (z, a)$, set $K=Q \cap E$, $J=Q\setminus E$, $L=E\setminus Q$,
and define
\begin{alignat}{2}
  f(y)    &=   s &&=  \norm{z}^2 - a ,                              \\
  g_x(y)  &=   \|z - x\|^2 - s - w_x &&=  - 2 \scalprod{z}{x} + a + \norm{x}^2 , \quad\forall x \in Q \cup E.
\end{alignat}
Noting that a point in $K = Q \cap E$ gives rise to two inequalities
that combine to an equality constraint,
we see that the thus defined optimization problem is equivalent
to the original optimization problem \eqref{eqn:G1}, \eqref{eqn:G2}, \eqref{eqn:G3}.
The gradients are
\begin{align}
  \nabla f(y)    &=  (2 z,   -1 ) ,    \\
  \nabla g_x(y)  &=  (  -2 x,   1 ) .
\end{align}
The stationarity condition is equivalent to writing $z$ as an affine combination
of the points in $Q\cup E$, with non-negative coefficients
for the points in $Q\setminus E$ and non-positive coefficients for the points in $E\setminus Q$.
The complementary slackness condition translates into vanishing coefficients
whenever the inequality is strict; that is: whenever the point does
not lie on the sphere.
We can now specialize the general conditions.

\begin{theorem}[Special KKT Conditions]
  \label{Special KKT Conditions}
  Let $S$ be a sphere that includes $Q \subseteq X$
  and excludes $E \subseteq X$.
  Then $S$ is the smallest such sphere iff its center
  is an affine combination of the points $x \in Q \cup E$, 
  \[ z = \sum \lambda_x x \quad\text{with}\quad 1 = \sum \lambda_x,\]
  such that
  \begin{enumerate}[(I)]
    \item $\lambda_x = 0$ whenever $x$ does not lie on $S$,
    \item $\lambda_x \geq 0$ whenever $x \in Q \setminus E$, and
    \item $\lambda_x \leq 0$ whenever $x \in E \setminus Q$.
  \end{enumerate}
\end{theorem}

\paragraph{Combinatorial formulation}
We are almost ready to prove the important technical result that the radius functions
are generalized discrete Morse functions provided the points
are in general position.
We begin by formalizing the latter condition.
A \emph{circumsphere} of a set $P \subseteq \Rspace^n$ is an $(n-1)$-sphere
such that all points of $P$ lie on the sphere.
A sufficient condition for the existence of a circumsphere is that
$P$ be affinely independent.
For sets of $n$ or fewer points, the circumsphere is not unique,
but by \theref{Special KKT Conditions}, there is only one \emph{smallest circumsphere},
namely the one whose center lies in the affine hull of the points.
We formulate conditions under which only the minimum number of points
lie on any smallest circumsphere.
\begin{definition}[General Position Assumption]
  \label{General Position Assumption}
  A finite set $X$ in $\Rspace^n$ is in \emph{general position}
  if for every $P \subseteq X$ of at most $n+1$ points
  \begin{enumerate}[(a)]
    \item $P$ is affinely independent, and
    \item no point of $X \setminus P$ lies on the smallest circumsphere
      of $P$.
  \end{enumerate}
\end{definition}
The formulation applies to points as well as to points with weights.
We therefore assume the more general case in which $X$ is a finite
set of weighted points in general position in $\Rspace^n$.
Let $S$ be an $(n-1)$-sphere,
write $\Incl{S}, \Excl{S} \subseteq X$ for the subsets of
included and excluded points,
and set $\On{S} = \Incl{S} \cap \Excl{S}$.
Now assume that $S$ is the smallest circumsphere of some set $P$,
that is: the center $z$ of $S$ lies in the affine hull of $P$,
and $P = \On{S}$ by general position.  We have
\[z = \sum_{x \in \On{S}} \rho_x x \quad\text{with}\quad 1 = \sum_{x \in \On{S}} \rho_x.\]
By general position, the affine combination is unique,
and $\rho_x \neq 0$ for all $x \in \On{S}$.
We call
\begin{align}
  \Front{S}  &=  \{ x \in \On{S}  \mid  \rho_x > 0 \} , \\
  \Back{S}   &=  \{ x \in \On{S}  \mid  \rho_x < 0 \} 
\end{align}
the \emph{front face} and the \emph{back face} of $\On{S}$, respectively.
We have $\Back{S} = \emptyset$ iff the circumcenter $z$ is contained
in the convex hull of $\On{S}$.
Using these definitions, we now give a combinatorial version
of the Karush--Kahn--Tucker conditions.
\begin{theorem}[Combinatorial KKT Conditions]
\label{Combinatorial KKT Conditions}
  Let $X$ be a finite set of weighted points in general position in $\Rspace^n$.
  Let $Q, E \subseteq X$ for which there exists a sphere $S$
  with $Q \subseteq \Incl{S}$ and $E \subseteq \Excl{S}$.
  It is the smallest such sphere, $S = \DSphere{Q}{E}$, iff
  \begin{enumerate}[(i)]
    \item $S$ is the smallest circumsphere of $\On{S}$,
    \item $\Front{S} \subseteq Q$, and
    \item $\Back{S} \subseteq E$.
  \end{enumerate}
\end{theorem}
\begin{proof}
  We first show that \theref{Special KKT Conditions} imply
  \theref{Combinatorial KKT Conditions}.
  The center of $S$ lies in the affine hull of 
  $P = \{x \in Q \cup E \mid \lambda_x \neq 0\}$,
  and Condition (I) implies that $P \subseteq \On{S}$
  which in turn implies~(i).
  By general position, we have
  $\On{S} = P$, and
  we can apply the definition of front and back face,
  letting $\rho_x = \lambda_x$ for all $x \in \On{S}$.
  Condition (II) now 
  says that $x \not\in E$ implies $\rho_x \geq 0$,
  or, equivalently, that $\rho_x < 0$ implies $x \in E$.
  Hence, $\Back{S} \subseteq E$.
  Similarly, Condition (III)
  yields $\Front{S} \subseteq Q$.
  
  We next show that Conditions (i) to (iii) imply (I) to (III).
  First note that (ii) and (iii) imply 
  $\On{S} \subseteq Q \cup E$.
  We define $\lambda_x = \rho_x$ for all $x \in \On{S}$,
  and $\lambda_x = 0$ for $x \in (Q \cup E) \setminus \On{S}$.
  Now (I) is satisfied by construction,
  and the inclusion $\Back{S} \subseteq E$ implies (II),
  while the inclusion $\Front{S} \subseteq Q$ implies (III).
\end{proof}

\paragraph{Partition into intervals}
Fix $E \subseteq X$ and recall that $\radfun{E}$ maps each simplex
$Q \in \sDelTri{X}{E}$ to the squared radius of $S = \DSphere{Q}{E}$.
This implies $\radfun{E} (P) = \radfun{E} (Q)$
for all $P \in [\Front{S}, \Incl{S}]$.
To prove that $\radfun{E}$ is a generalized discrete Morse function,
it remains to show that $\radfun{E} (P) < \radfun{E} (Q)$ whenever
$P \subseteq Q$ do not belong to the same interval.
But this is clear from \theref{General Position Assumption}.
\begin{theorem}[Selective Delaunay Morse Function Theorem]
  \label{Selective Delaunay Morse Function Theorem}
  Let $X$ be a finite set of weighted points in general position in $\Rspace^n$,
  and $E \subseteq X$.
  Then the radius function, $\radfun{E} \colon \sDelTri{X}{E} \to \Rspace$,
  is a generalized discrete Morse function
  whose discrete gradient consists of the intervals $[\Front{S}, \Incl{S}]$
  over all Delaunay spheres $S = \DSphere{Q}{E}$ with $Q \in \sDelTri{X}{E}$.
\end{theorem}
Setting $E = \emptyset$, we have $\Front{S} = \On{S}$
because $\Back{S} = \emptyset$ by \theref{Combinatorial KKT Conditions}.
Symmetrically, for $E = X$, we have $\Incl{S} = \On{S}$.
This implies the following two special cases of the above theorem.
\begin{corollary}[\v Cech Morse Function Corollary]
  The \v Cech radius function of a finite set of weighted points
  in general position is a generalized discrete Morse function.
  Its gradient consists of the intervals $[\On{S}, \Incl{S}]$
  over all \v Cech spheres~$S$ of~$X$.
\end{corollary}
\begin{corollary}[Delaunay Morse Function Corollary]
  The Delaunay radius function of a finite set of weighted points
  in general position is a generalized discrete Morse function.
  Its gradient consists of the intervals $[\Front{S}, \On{S}]$
  over all Delaunay spheres~$S$ of~$X$.
\end{corollary}
Another straightforward consequence of \theref{Combinatorial KKT Conditions}
is the invariance of the critical simplices.
To state this theorem, we call $Q \in \DelTri{X}$ 
a \emph{centered Delaunay simplex} if 
the center of $S = \DSphere{Q}{X}$ is contained 
in the convex hull of $Q$.
Equivalently, we have
$\Front{S} = \On{S} = \Incl{S}$.
Note that in this case $S = \DSphere{Q}{E}$ for all sets $E \subseteq X$.
\begin{corollary}[Critical Simplex Corollary]
  \label{Critical Simplex Corollary}
  Let $X$ be a finite set of weighted points in general position in $\Rspace^n$.
  Independent of $E$, a subset $Q \subseteq X$ is a critical simplex of $\radfun{E}$
  iff $\Dsqradius{Q}{\emptyset} = \Dsqradius{Q}{E} = \Dsqradius{Q}{X}$
  iff~$Q$~is a centered Delaunay simplex.
\end{corollary}

\paragraph{Wrap complex}
It is now easy to define the Wrap complex using the gradient $V_X$ of
the Delaunay radius function $\radfun{X} \colon \DelTri{X} \to \Rspace$,
which partitions the Delaunay triangulation into intervals
of the form $[\Front{S}, \On{S}]$.
Let $\Digraph$ be the directed graph whose nodes are the intervals
in $V_X$, with an arc from $\mu$ to $\nu$ if there are simplices
$P \in \mu$ and $Q \in \nu$ with $P \subseteq Q$.
It defines a partial order on $V_X$.
The \emph{lower set} of a subset $A \subseteq V_X$,
denoted by $\dsc A$, is the collection of intervals from which $A$
can be reached along directed paths in $\Digraph$.
The lower set of a singular interval is akin to the stable manifold
of a critical point in smooth Morse theory,
except that the lower sets of the critical simplices do not
necessarily form a partition.
Indeed, the lower sets can overlap,
and some of the simplices may not belong to the lower set of any
critical simplex.
The latter can be considered to belong to the lower set of the
`outside', but it will not be necessary to formalize this intuition.
The \emph{Wrap complex} for $r \geq 0$ consists of all simplices
in the lower set of the singular intervals with a Delaunay sphere
of radius at most $r$:
\begin{align}
  \Singular{r}{X}  &=  \{ [Q,Q] \in V_X  \mid  \radfun{X} (Q) \leq r^2 \}, \\
  \Wrap{r}{X}      &=  \bigcup \dsc{\Singular{r}{X}} .
\end{align}
The original definition of the Wrap complex \cite[Section 6]{Edelsbrunner2003Surface}
corresponds to $\Wrap{\infty}{X}$, which we simply denote as $\Wrap{}{X}$.
In the terminology of \cite{Edelsbrunner2003Surface},
a \emph{confident} simplex is the upper bound of a non-singular
Delaunay interval, while all other simplices in the
interval are \emph{equivocal}.
The critical Delaunay simplices are called \emph{centered},
in accordance with \theref{Critical Simplex Corollary}.
Clearly, $\Wrap{r}{X} \subseteq \Wrap{t}{X}$ whenever $r \leq t$.
Moreover, from the construction as a union of lower sets we immediately have $\Wrap{r}{X} \subseteq \Del{r}{X}$
for every $r \in \Rspace$.

\section{Simple-Homotopy Equivalence}
\label{sec5}

In this section, we prove that the various complexes considered
in this paper are simple-homotopy equivalent.
Throughout, we write $Q - x = Q \setminus \{x\}$
and $Q + x = Q \cup \{x\}$,
noting that one of these two simplices is equal to $Q$.
The proof strategy is based on the construction of two discrete gradients.
The first one is defined on the full simplex on $X$
and induces the simplicial collapse $\Cech{r}{X} \searrow \DCech{r}{X}$
by removing all non-Delaunay simplices.
The second discrete gradient is defined on the Delaunay triangulation
$\DelTri{X}$ and induces the collapse $\DCech{r}{X} \searrow \Del{r}{X}$.
While sketched here for the collapse of the \v Cech to the Delaunay complex,
the construction more generally establishes a collapse
of selective Delaunay complexes $\sDel{r}{X}{E}
\searrow \sDel{r}{X}{F}$ for $E \subseteq F \subseteq X$.

The discrete gradients are constructed by assigning to each collapsed simplex
$Q \in \sDel{r}{X}{F} \setminus \sDel{r}{X}{E}$ a point $x \in F \setminus E$ 
that turns the sphere $\DSphere{Q}{E}$
infeasible for the excluded set $F$.
As a consequence, the sphere $\DSphere{Q}{F}$ will either
have a larger radius or not exist at all.
The choice of the vertex~$x$ will be such that both simplices
$Q-x$ and $Q+x$ will be assigned the same vertex~$x$.
The resulting collections of pairs $\{Q-x,Q+x\}$ are then verified
to be the pairs of a discrete gradient that induces the desired collapses.

\paragraph{Auxiliary lemmas}
Before we begin, we present several simple lemmas,
which will be useful in our proofs.
First, we characterize a special class of discrete gradients
that are constructed using a distinguished vertex.
\begin{lemma}[Vertex Gradient Lemma]
  \label{Vertex Gradient Lemma}
  Let $K$ be a simplicial complex, $V$ a discrete vector field on $K$,
  and $x$ a vertex of $K$.
  If every pair in $V$ can be written as $\{Q - x, Q + x\}$
  for some simplex $Q$, then $V$ is a discrete gradient.
\end{lemma}
\begin{proof}
  Consider the function $f \colon K \to \Rspace$ defined by 
  taking average dimensions:
  \begin{align}
    f(Q)  &=  \begin{cases}
                \frac{1}{2} \big(\dim(Q - x)+\dim(Q + x)\big)
                     & \text{if } \{Q - x, Q + x\} \in V, \\
                \dim(Q)& \text{otherwise}.
              \end{cases}
  \end{align}
  Clearly $f(Q - x) = f(Q + x)$ whenever $\{Q-x, Q+x\} \in V$.
  Suppose now that $P \subseteq Q$ with $\dime{P} = \dime{Q} - 1$.
  If $P$ is critical or $Q$ is critical, then $f(P) < f(Q)$ is easy to see.
  Assume therefore that $\{P-x, P+x\}$ and $\{Q-x, Q+x\}$ are different
  pairs in $V$.
  Then $f(P) < f(Q)$ unless $P = P-x$ and $Q = Q+x$, but the latter case
  would imply $Q = P+x$, which contradicts the disjointness of the pairs.
\end{proof}

We will also make use of the following lemma, which allows us to
extend a discrete gradient on a subcomplex by a discrete gradient
on its complement.
A proof can be found in \cite[Lemma~4.3]{Jonsson2008Simplicial}.
\begin{lemma}[Gradient Composition Lemma]
  \label{Gradient Composition Lemma}
  Let $K \subseteq L$ be simplicial complexes
  with discrete gradients $V$ of $K$ and $W$ of $L$.
  If every pair in $W$ is disjoint from $K$, then
  the pairs in $V \cup W$ define a discrete gradient on $L$.
\end{lemma}

The following lemma will be useful to obtain a common refinement
of two generalized gradients by taking the sum of the two corresponding
generalized Morse functions.
We omit the proof, which is straightforward.
\begin{lemma}[Sum Refinement Lemma]
  \label{Sum Refinement Lemma}
  Let $f \colon K \to \Rspace$ and $g \colon L \to \Rspace$
  be generalized discrete Morse functions with gradients $V$ and $W$.
  Then $f+g \colon K \cap L \to \Rspace$ is a generalized Morse function
  with gradient
  $\{I \cap J \mid I \in V, J \in W, I \cap J \neq \emptyset \}$.
\end{lemma}

In order to analyze the discrete gradients of radius functions, 
we note that for $Q, E \subseteq X$ and $S = \DSphere{Q}{E}$,
removing a point $x \in \Incl{S}$ from $Q$
affects the smallest sphere only if $x \in \Front S$.
Likewise, removing a point $y \in \Excl S$ from $E$
affects the smallest sphere only if $y \in \Back S$.
\begin{lemma}[Same Sphere Lemma]
  \label{Same Sphere Lemma}
  Let $Q$ be a simplex in $\sDelTri{X}{E}$ and $S = \DSphere{Q}{E}$ the
  smallest sphere that includes $Q$ and excludes $E$.
  Then
  \begin{enumerate}[(i)]
    \item $S = \DSphere{Q - x}{E} = \DSphere{Q + x}{E}$ iff
      $x \in \Incl{S} \setminus \Front S$,
    \item $S = \DSphere{Q}{E - y} = \DSphere{Q}{E + y}$ iff
      $y \in \Excl S \setminus \Back S$.
  \end{enumerate}
\end{lemma}
\begin{proof}
  We show (i), omitting the proof of (ii), which is analogous.
  By \theref{Combinatorial KKT Conditions}, we have 
  $\Front{S} \subseteq Q \subseteq \Incl{S}$.
  Now $x \in \Incl{S} \setminus \Front S$ is equivalent to 
  $\Front{S} \subseteq Q-x$ and $Q+x \subseteq \Incl{S}$,
  which implies that $Q-x$ and $Q+x$ belong to the same interval
  $[\Front{S}, \Incl{S}]$.
  The claim follows.
\end{proof}

\paragraph{Pairing lemmas}
Next, we prove two key technical lemmas that will facilitate
the construction of discrete gradients proving our collapsibility results.
Let $E \subseteq F \subseteq X$ and consider
a simplex $Q$ whose spheres $\DSphere{Q}{E}$ and $\DSphere{Q}{F}$
both exists but are different.
We show that there is a point in $F \setminus E$ such that adding
the point to $Q$ or removing it from $Q$ affects neither of the spheres.
\begin{lemma}[First Simplex Pairing Lemma]
  \label{First Simplex Pairing Lemma}
  Let $E \subseteq F \subseteq X$ and $Q \in \sDelTri{X}{F}$ with
  $\DSphere{Q}{E} \neq \DSphere{Q}{F}$.
  Then there exists a point $x \in F \setminus E$ such that
  \begin{enumerate}[(i)]
    \item $\DSphere{Q-x}{E} = \DSphere{Q+x}{E}$,
    \item $\DSphere{Q-x}{F} = \DSphere{Q+x}{F}$.
  \end{enumerate}
\end{lemma}
\begin{proof}
  To construct the point in question, we write $S = \DSphere{Q}{E}$
  and $T = \DSphere{Q}{F}$.
  By \theref{Combinatorial KKT Conditions},
  we have $T = \DSphere{Q}{\Back T}$.
  By assumption we have $S \neq T$, and because $E \subseteq F$, the sphere $S$ is smaller than~$T$.
  It can therefore not exclude all points of $\Back{T}$,
  and we let $x$ be any point in $\Back{T} \setminus \Excl{S}$.
  Clearly $x \in F \setminus E$.
  Finally, we apply the first claim in \theref{Same Sphere Lemma} twice
  to get the claimed relations.
  First we use $x \not\in \Excl{S}$ to get $x \in \Incl{S} \setminus \Front{S}$
  so applying the lemma gives $\DSphere{Q-x}{E} = \DSphere{Q+x}{E}$ as claimed in (i).
  Second we use $x \not\in \Front{T}$ to get $x \in \Incl{T} \setminus \Front{T}$
  so applying the lemma gives $\DSphere{Q-x}{F} = \DSphere{Q+x}{F}$ as claimed in (ii).
\end{proof}

We note that the set $\Back{T} \setminus \Excl{S}$
and the point $x$ in this set selected for $Q$ in the above proof
work for both $Q-x$ and for $Q+x$.
In other words, substituting $Q-x$ or $Q+x$ for $Q$ in
\theref{First Simplex Pairing Lemma} does not affect
the claimed relations, and we can consistently select
the same point $x$ for both $Q-x$ and $Q+x$.
The following lemma makes this observation precise.
\begin{lemma}[First Consistent Pairing Lemma]
  \label{First Consistent Pairing Lemma}
  Assuming $E$ is a proper subset of~$F$, there is a map
  $$
    \varphi \colon \{Q \in \sDelTri{X}{F} \mid 
  \DSphere{Q}{E} \neq \DSphere{Q}{F} \} \to F \setminus E
  $$
  such that $x = \varphi (Q)$ satisfies the properties of
  \theref{First Simplex Pairing Lemma},
  and $x = \varphi (Q-x) = \varphi (Q+x)$.
\end{lemma}
\begin{proof}
  To define the map $\varphi$, we let $x_1, x_2, \ldots, x_m$
  be an arbitrary but fixed ordering of the vertices in $X$.
  Let $Q \in \sDel{}{X}{F}$ be a simplex, and let
  $S = \DSphere{Q}{E}$ and $T = \DSphere{Q}{F}$.
  Now consider the vertex
  $x=x_i \in \Back{T} \setminus \Excl{S}$ 
  with the smallest index in the chosen ordering.
  This vertex $x$ satisfies the properties of
  \theref{First Simplex Pairing Lemma}, as shown in its proof.
  The choice of this vertex depends only on the two spheres 
  $S = \DSphere{Q}{E}$ and $T = \DSphere{Q}{F}$, 
  and by \theref{First Simplex Pairing Lemma},
  using either $Q-x$ or $Q+x$ in place of $Q$ yield this same pair of spheres.
  Defining $\varphi (Q) = x$,
  we conclude $x = \varphi (Q-x) = \varphi (Q+x)$.
\end{proof}

Consider next a simplex that belongs to the selective Delaunay complex
for $E$ but not for $F$.
We show that there exists a point $x$ that has properties similar
to those established in \theref{First Simplex Pairing Lemma}.
\begin{lemma}[Second Simplex Pairing Lemma]
  \label{Second Simplex Pairing Lemma}
  Let $E \subseteq F \subseteq X$ and let $Q$ be a simplex in $\sDelTri{X}{E}$
  but not in $\sDelTri{X}{F}$.
  Then there exists a point $x \in F \setminus E$ such that
  \begin{enumerate}[(i)]
    \item $\DSphere{Q-x}{E} = \DSphere{Q+x}{E}$,
    \item both $Q-x$ and $Q+x$ are not in $\sDelTri{X}{F}$.
  \end{enumerate}
\end{lemma}
\begin{proof}
  To construct the point in question, we write
  $F_Q = F \cap \Excl{\DSphere{Q}{E}}$ and note that
  $S = \DSphere{Q}{E} = \DSphere{Q}{F_Q}$
  by \theref{Combinatorial KKT Conditions}.
  In particular, $Q \in \sDel{}{X}{F_Q}$.
  Let $A \subseteq F$ be a subset of the points $F$ and $x \not\in A$
  satisfying $F_Q \subseteq A \subseteq A+x \subseteq F$
  such that $Q$ belongs to $\sDel{}{X}{A}$ but not to $\sDel{}{X}{A+x}$.
  It is clear that such $A$ and $x$ exist.
  Since $x \not\in \Excl{S}$,
  we have $x \in \Incl{S} \setminus \On{S} \subseteq \Incl{S} \setminus \Front{S}$.
  Applying \theref{Same Sphere Lemma}, we get
  $S = \DSphere{Q-x}{E} = \DSphere{Q+x}{E}$, as claimed in (i).
  Similarly, 
  $Q \in \sDel{}{X}{A} \setminus \sDel{}{X}{A+x}$ implies that 
  $x \not\in \Excl{\DSphere{Q}{A}}$, and we also get $\DSphere{Q-x}{A} = \DSphere{Q+x}{A}$,
  which will be useful shortly.
  In particular, note that $Q-x, Q+x$ both belong to $\sDel{}{X}{A}$.

  To see (ii), we note that $Q \not\in \sDel{}{X}{A+x}$ by assumption
  and therefore also $Q+x \not\in \sDel{}{X}{A+x}$.
  Since $A+x \subseteq F$, this implies the second relation in (ii).
  To derive a contradiction, we assume $Q-x \in \sDel{}{X}{F}$ or,
  by implication, $Q-x \in \sDel{}{X}{A+x}$.
  Using $\DSphere{Q-x}{A} = \DSphere{Q+x}{A}$, we get $x \not\in \Excl{\DSphere{Q-x}{A}}$.
  Hence, $x$ lies inside the sphere $\DSphere{Q-x}{A}$
  and we have $\DSphere{Q-x}{A+x} \neq \DSphere{Q-x}{A}$.
  On the other hand, since $Q+x \not\in \sDel{}{X}{A+x}$,
  we know that $x \not\in \Incl{\DSphere{Q-x}{A+x}}$.
  Applying the second claim in \theref{Same Sphere Lemma},
  we get $\DSphere{Q-x}{A+x}=\DSphere{Q-x}{A}$,
  a contradiction to the above.
  We thus conclude that $Q-x, Q+x$ both do not belong to
  $\sDel{}{X}{A+x}$ and therefore also not to $\sDel{}{X}{F}$.
\end{proof}

Similar to above, the set $A$ and the point $x \not\in A$ selected for $Q$
in the above proof work for both $Q-x$ and for $Q+x$, and we can consistently 
select the same point $x$ for both $Q-x$ and $Q+x$.
\begin{lemma}[Second Consistent Pairing Lemma]
  \label{Consistent Pairing Lemma}
  Assuming $E$ is a proper subset of~$F$, there is a map
  $$
    \psi \colon \sDelTri{X}{E} \setminus \sDelTri{X}{F} \to F \setminus E
  $$
  such that $x = \psi (Q)$ satisfies the properties of
  \theref{Second Simplex Pairing Lemma}, and $x = \psi (Q-x) = \psi (Q+x)$.
\end{lemma}
\begin{proof}
  To define $\psi$, we let $x_1, x_2, \ldots, x_m$
  be an arbitrary but fixed ordering of the vertices in $X$.
  Let $Q \in \sDel{}{X}{E}$ be a simplex considered in
  \theref{Second Simplex Pairing Lemma} and recall that
  $F_Q = F \cap \Excl{\DSphere{Q}{E}}$.
  There is a unique index~$j$ such that
  $A = F_Q \cup (F \cap \{x_1, x_2, \ldots, x_{j-1}\})$ and $x = x_j$
  satisfy the criteria spelled out in the proof of
  \theref{Second Simplex Pairing Lemma};
  concretely, we have
  $Q \in \sDel{}{X}{A} \setminus \sDel{}{X}{A+x}$ and $x \in F \setminus E$.
  We use this choice of vertex to define $\psi (Q) = x_j$.

  Note that the choice of $F_Q$ depends only on $\DSphere{Q}{E}$,
  and since $\DSphere{Q-x}{E} = \DSphere{Q+x}{E}$,
  we get the same $A$ and $x$ for $Q-x$ as well as for $Q+x$.
  We also have $Q-x, Q+x \in \sDel{}{X}{A} \setminus \sDel{}{X}{A+x}$,
  as pointed out in the proof of \theref{Second Simplex Pairing Lemma},
  so we conclude that $x = \psi (Q-x) = \psi (Q+x)$, as claimed.
\end{proof}

\paragraph{Collapsing}
We are now ready to prove two collapsibility results for selective
Delaunay complexes.
They will imply the main results of this paper.
\begin{theorem}[Selective Delaunay Collapsing Theorem]
  \label{Selective Delaunay Collapsing Theorem}
  Let $X$ be a finite set of possibly weighted points
  in general position in $\Rspace^n$, and let $E \subseteq F \subseteq X$.
  Then
  \begin{align}
    \sDel{r}{X}{E} \searrow \sDel{r}{X}{E} \cap \sDelTri{X}{F}
                   \searrow \sDel{r}{X}{F}
  \end{align}
  for every $r \in \Rspace$.
\end{theorem}
\begin{proof}
  We show that both collapses are induced by discrete gradients
  constructed with the help of the two Simplex Pairing Lemmas.
  We first prove the second collapse,
  $\sDel{r}{X}{E} \cap \sDelTri{X}{F} \searrow \sDel{r}{X}{F}$.
  Let $V_E$ and $V_F$ be the generalized discrete gradients
  of the radius functions $\radfun{E} \colon \sDel{}{X}{E} \to \Rspace$
  and $\radfun{F} \colon \sDel{}{X}{F}$.
  By \theref{Sum Refinement Lemma}, the function
  $\radfun{E} + \radfun{F} \colon \Del{}{X,F} \to \Rspace$
  is a generalized discrete Morse function,
  and its generalized discrete gradient is
  \begin{align}
    W  &=  \{I \cap J \mid I \in V_E, J \in V_F, I \cap J \neq \emptyset \}.
  \end{align}
  For any simplex $Q$ that belongs to $\Del{r}{X,E} \cap \Del{}{X,F}$
  but not to $\Del{r}{X,F}$, the~sphere $\DSphere{Q}{E}$ has radius
  at most $r$ but $\DSphere{Q}{F}$ has radius larger than $r$.
  The set of such simplices is partitioned by a subset of the intervals
  in $W$.
  Since in particular $\DSphere{Q}{E} \neq \DSphere{Q}{F}$,
  \theref{First Simplex Pairing Lemma} implies that
  this partition contains no singular intervals.
  By \theref{Generalized Gradient Collapsing Theorem}, $W$ induces the collapse
  $\Del{r}{X,E} \cap \Del{}{X,F} \searrow \Del{r}{X,F}$.
  Note that the pairs
  \begin{align}
    P_0 = \{ \{Q-x,Q+x\} \mid Q \subseteq X,\; x \in X,\; \varphi (Q) = x \}
  \end{align}
  defined using the map $\varphi$ from \theref{First Simplex Pairing Lemma}
  yield a vertex refinement of the generalized gradient $W$.

  We next prove the first collapse, 
  $\sDel{r}{X}{E} \searrow \Del{r}{X,E} \cap \Del{}{X,F}$,
  using the pairs obtained from \theref{Second Simplex Pairing Lemma}
  to partition the complement,
  $\sDel{}{X}{E} \setminus \sDel{}{X}{F}$.
  Revisiting the construction of $\psi$,
  we fix a total order $x_1, x_2, \ldots, x_m$ on $X$ and define
  \begin{align}
    K_0 &= \Del{}{X,E},\\
    K_i &= K_{i-1} \setminus \{Q \subseteq X \mid \psi (Q) = x_i\}.
  \end{align}
  We thus get a filtration
  $K_m \subseteq K_{m-1} \subseteq \ldots \subseteq K_0$
  that starts with $\sDel{}{X}{F}$ and ends with $\sDel{}{X}{E}$.
  By \theref{Vertex Gradient Lemma}, the pairs
  \begin{align}
    P_i  &=  \{ \{Q-x_i,Q+x_i\} \mid Q \subseteq X,\; \psi (Q) = x_i\} ,
             \quad i=1,\dots,m,
  \end{align}
  give rise to a discrete gradient, $V_i$,
  and since $V_i$ partitions $K_i \setminus K_{i+1}$ into pairs,
  it induces a collapse $K_i \searrow K_{i+1}$
  by \theref{Gradient Collapsing Theorem}. 
  The union of all such sets of pairs,
  \begin{align}
    \bigcup_{i=1}^m P_i  &=  \{ \{Q-x,Q+x\} \mid Q \subseteq X,\;
    x \in X,\; v(Q) = x \} ,
  \end{align}
  forms a partition of $\sDel{}{X}{E} \setminus \sDel{}{X}{F}$
  and, by applying \theref{Gradient Composition Lemma} inductively,
  yields a gradient on $\sDel{}{X}{E}$
  inducing the collapse $\sDel{}{X}{E} \searrow \sDel{}{X}{F}$.
\end{proof}
  
We remark that the pairs of \theref{First Simplex Pairing Lemma} and 
\theref{Second Simplex Pairing Lemma} can be combined according to 
\theref{Gradient Composition Lemma}.
The result is a single discrete gradient $V$,
with the pairs $\bigcup_{i=0}^m P_i$, which induce both collapses,
for any choice of $r$ simultaneously.

\paragraph{The collapsing sequence}
We are now ready to state and prove the collapsing sequence,
according to which the \v{C}ech, the \v{C}ech--Delaunay,
the Delaunay, and the Wrap complexes -- all for the same parameter
$r \in \Rspace$ -- are simple-homotopy equivalent.
The relation between the first three complexes follows directly
from \theref{Selective Delaunay Collapsing Theorem},
and it is not difficult to expand the relation to include the Wrap complex.
\begin{theorem}[\v Cech--Delaunay Collapsing Theorem]
\label{Cech--Delaunay Collapsing Theorem}
  Let $X$ be a finite set of possibly weighted points in general position
  in $\Rspace^n$.
  Then
  \begin{align}
    \Cech{r}{X} \searrow \DCech{r}{X} \searrow \Del{r}{X} \searrow \Wrap{r}{X}
  \end{align}
  for every $r \in \Rspace$.
\end{theorem}
\begin{proof}
  Setting $E = \emptyset$ and $F = X$ in
  \theref{Selective Delaunay Collapsing Theorem}, we get the first
  two relations.
  To see the third, recall that $\Wrap{r}{X} \subseteq \Del{r}{X}$
  for every $r \in \Rspace$.
  To show that the latter complex collapses onto the former complex,
  we consider the gradient $V_X$ of the Delaunay radius function
  $\radfun{X} \colon \DelTri{X} \to \Rspace$.
  Each interval in $V_X$ is either disjoint of $\Del{r}{X}$
  or a subset of $\Del{r}{X}$, and similarly for $\Wrap{r}{X}$.
  Moreover, by construction all critical simplices of $\Del{r}{X}$
  are also contained in $\Wrap{r}{X}$.
  It follows that $\Del{r}{X} \setminus \Wrap{r}{X}$ is the
  disjoint union of non-singular intervals in $V_X$,
  and \theref{Generalized Gradient Collapsing Theorem}
  implies $\Del{r}{X} \searrow \Wrap{r}{X}$.
\end{proof}

Independent of $r$, all collapses are induced by the same discrete
gradient constructed in \theref{Selective Delaunay Collapsing Theorem}. 

\section{Consequences}
\label{sec6}

Next, we discuss two structural results implied by our collapsing sequence.
The first concerns the persistent homology of the filtrations obtained
by letting $s = r^2$ range over $\Rspace$,
while the second uses selective Delaunay complexes to compare
Delaunay complexes of different point sets.

\paragraph{Naturality and persistence}
Regarding a filtration as a diagram of topological spaces connected by inclusions,
a \emph{natural transformation} from a filtration $(K_t)_{t \in \Rspace}$
to another filtration $(L_t)_{t \in \Rspace}$
is a family of continuous maps $K_t \to L_t$ such that the diagram 
\[
\begin{tikzcd}[column sep=4.5ex]
  K_r \arrow{r} \arrow[hook]{d} & L_r \arrow[hook]{d}\\
  K_t \arrow{r} & L_t
\end{tikzcd}
\]
commutes for all $r \leq t$.
The \emph{persistent homology} of $(K_t)_{t \in \Rspace}$
is the diagram of homology groups $H_*(K_t)$ connected by the homomorphisms
induced by the inclusions $K_r \hookrightarrow K_t$ for $r \leq t$.
Since homology is a functor, it sends a natural transformation
of filtrations to a natural transformation of their persistent homology.
By \theref{Cech--Delaunay Collapsing Theorem}, the diagram
$$
\begin{tikzcd}[arrows=hook,column sep=4.5ex]
  \Wrap{r}{X} \arrow {r}\arrow{d} & \Del{r}{X} \arrow {r}\arrow{d} & \DCech{r}{X} \arrow{r}\arrow{d} & \Cech{r}{X} \arrow{d}\\
  \Wrap{t}{X} \arrow {r} & \Del{t}{X} \arrow {r} & \DCech{t}{X} \arrow {r} & \Cech{t}{X}
\end{tikzcd}
$$
commutes for all $r \leq t$.
The horizontal inclusion maps in this diagram correspond to the collapses 
of \theref{Cech--Delaunay Collapsing Theorem}.
This means that for any two of the four filtrations,
the inclusion maps constitute a natural transformation,
which is a simple-homotopy equivalence at each filtration index.
As a consequence,
we have the following implication on the persistent homology of the filtrations.
\begin{corollary}[Persistence Isomorphism Corollary]
  The \v Cech, Delaunay--\v Cech, Delaunay, and Wrap filtrations 
  have isomorphic persistent homology.
\end{corollary}
It should be clear that this corollary extends to the filtrations of complexes
considered in \theref{Selective Delaunay Collapsing Theorem}.

\subsection{Connecting Delaunay complexes for different point sets}
We now describe an application of \theref{Selective Delaunay Collapsing Theorem}
that highlights the selective Delaunay complexes as interesting objects
in their own right.

Assume we are given two finite point sets $X$ and $Y$ in $\Rspace^n$,
and denote the corresponding unions of balls of radius $r$
by $\Ball{r}{X}$ and $\Ball{r}{Y}$.
We think of $X$ as a geometric approximation of $Y$,
or of both as different approximations of some compact space.
We are interested in the homomorphisms in homology induced
by the inclusions $\Ball{r}{X} , \Ball{r}{Y}\hookrightarrow \Ball{r}{X \cup Y}$.
Ideally, the induced homomorphisms are isomorphisms for an appropriate choice
of $r$.
However, this cannot always be achieved and no such choice of $r$ may exists, 
even if the Hausdorff distance $\delta=d_H(X,Y)$ is small.
Nevertheless, the induced homomorphisms constitute natural transformations 
that can be thought of as approximate isomorphisms, up to a shift by $\delta$ 
in the index~$r$ of the diagrams.
Known as \emph{$\delta$-interleavings} \cite{Chazal2009Proximity,Bauer2015Induced},
they translate the geometric closeness of $X$ and $Y$ into a structural similarity
of their persistent homology.
The persistent homology is described uniquely up to isomorphism
by a collection of intervals, called the \emph{persistence barcode},
and indeed the mentioned homomorphisms further induce a matching
between the persistence barcodes of $(\Ball{r}{X})_{r \in \Rspace}$
and $(\Ball{r}{Y})_{r \in \Rspace}$ that makes this similarity
explicit \cite{Bauer2015Induced}.
The relevance of interleavings motivates the interest in the above
inclusion maps and their homology.

To construct the homology of these inclusion maps combinatorially,
we observe that the maps
$\Ball{r}{X}, \Ball{r}{Y} \hookrightarrow \Ball{r}{X \cup Y}$
can be described up to isomorphisms on the level of \v{C}ech complexes,
as both $\Cech{r}{X}$ and $\Cech{r}{Y}$ are subcomplexes
of $\Cech{r}{X \cup Y}$; see the diagram below.
The situation is different for Delaunay complexes because there is
no canonical simplicial map $\Del{r}{X} \to \Del{r}{X \cup Y}$
corresponding to the inclusion $X \hookrightarrow {X \cup Y}$.
To cope, we use selective Delaunay complexes
and construct a zigzag of inclusions connecting the two Delaunay complexes
$\Del{r}{X}$ and $\Del{r}{Y}$;
see the last two rows of the following diagram.

\begin{figure}[h]
\begin{tikzcd}[row sep=6ex,column sep=12ex]
\Cech{r}{X} \ar[right hook->]{rr} \ar[<-left hook,swap]{dd}{\simeq} &[between origins] &[between origins] \Cech{r}{X \cup Y}  &[between origins] &[between origins] \Cech{r}{Y} \ar[left hook->]{ll} \ar[<-right hook]{dd}{\simeq} \\
&[between origins] \sDel{r}{X \cup Y}{X} \ar[right hook->]{ur}{\simeq} &[between origins] &[between origins] \sDel{r}{X \cup Y}{Y} \ar[left hook->,swap]{ul}{\simeq} \\
\Del{r}{X} \ar[right hook->]{ur} &[between origins] &[between origins] \Del{r}{X \cup Y} \ar[left hook->,swap]{ul}{\simeq} \ar[right hook->]{ur}{\simeq} &[between origins] &[between origins] \Del{r}{Y} \ar[left hook->]{ul} 
\end{tikzcd}
\end{figure}

First note that $\Del{r}{X} = \sDel{r}{X}{X}$ is a subcomplex
of $\sDel{r}{X \cup Y}{X}$,
because adding $Y$ does not impose any constraints on how the points
in $X$ connect to each other.
Moreover, $\Del{r}{X \cup Y} = \sDel{r}{X \cup Y}{X \cup Y}$
is a subcomplex of $\sDel{r}{X \cup Y}{X}$,
because the former has a larger set of excluded points than the latter.
By \theref{Selective Delaunay Collapsing Theorem},
$\sDel{r}{X \cup Y}{X}$ collapses to $\Del{r}{X \cup Y}$.
Similar relations hold if we swap~$X$~and~$Y$.

From this diagram, one can see that the inclusions
$\Cech{r}{X} \hookrightarrow \Cech{r}{X \cup Y}$
are naturally homotopy-equivalent to the inclusions
$\Del{r}{X} \hookrightarrow \Del{r}{X \cup Y,X}$.
In particular, the natural transformation
$H_*(\Cech{r}{X} \hookrightarrow \Cech{r}{X \cup Y})$ is isomorphic to
$H_*(\Del{r}{X} \hookrightarrow \sDel{r}{X \cup Y}{X})$.
Furthermore, the discrete gradient from
\theref{Selective Delaunay Collapsing Theorem} inducing the collapse
$\sDel{r}{X \cup Y}{X} \searrow \Del{r}{X \cup Y}$
can be used to construct the induced isomorphism
$H_*(\sDel{r}{X \cup Y}{X}) \to H_*(\Del{r}{X \cup Y})$ on the level of cycles,
see \cite[Sections 7 and 8]{Forman1998Morse}.

\section{Discussion}
\label{sec7}

The main result of this paper is the construction of collapses
from the \v{C}ech to the Delaunay--\v{C}ech to the Delaunay and
finally to the Wrap complex of a finite set of possibly weighted
points in general position in $\Rspace^n$.
This is achieved by finding a common refinement of the generalized
discrete gradients of the \v{C}ech and Delaunay radius functions
into a generalized discrete gradient with the same set of
critical simplices.
While the collapse of the Delaunay--\v{C}ech to the Delaunay complex
is induced by a canonical generalized gradient, the construction of
the collapse from the \v{C}ech to the Delaunay--\v{C}ech complex
required the choice of a total order.
Is there an alternative proof that does not rely on such a choice?
The hope is that a natural such construction would reveal some
of the constraints on the collapse imposed by the geometry of the data.

We remark that the proof in this paper makes essential use of the
general position assumption.
Can the generalized discrete Morse theory be further generalized
so that this assumption is no longer necessary?

\bibliographystyle{abbrvnaturl}
\bibliography{collapsing-delaunay}

\todos

\end{document}